\documentclass{ws-procs9x6}

\usepackage{enumerate}
\usepackage{eucal}
\usepackage{bm}
\usepackage[all]{xy}
\usepackage{graphicx}

\newcommand{\N}{\mathbb{N}}

\newcommand{\Prob}{\mathbf{P}}
\newcommand{\Expec}{\mathbf{E}}
\newcommand{\Varia}{\mathbf{Var}}

\newcommand{\ee}{\mathrm{e}}

\newcommand{\etal}{\emph{et al}}
\newcommand{\ignore}[1]{}

\begin{document}

\DeclareGraphicsExtensions{.eps}

\title{Viral Evolution and Adaptation as a Multivariate Branching 
       Process\footnote{\uppercase{T}his work is supported by 
       \uppercase{CAPES}, \uppercase{CNP}q and \uppercase{FAPESP}.}}

\author{FERNANDO ANTONELI and FRANCISCO BOSCO\footnote{\uppercase{I}n
         memory of \uppercase{F}rancisco \uppercase{A. R. B}osco (1955--2012).}}
\address{Departamento de Inform\'atica em Sa\'ude and \\
         Laborat\'orio de Gen\^omica Evolutiva e Biocomplexidade\\
         Universidade Federal de S\~ao Paulo, S\~ao Paulo, SP, Brazil. \\
         Email: fernando.antoneli@unifesp.br, fbosco@unifesp.br}

\author{DIOGO CASTRO and LUIZ MARIO JANINI}
\address{Departamento de Microbiologia, Imunologia e Parasitologia and \\
         Departamento de Medicina, \\
         Universidade Federal de S\~ao Paulo, S\~ao Paulo, SP, Brazil.\\
         Email: diogo.castro@unifesp.br, janini@unifesp.br}

\footnotetext{\uppercase{U}nabridged version of
\emph{Viral Evolution and Adaptation as a Multivariate Branching Process}\cite{ABCJ11},
in BIOMAT 2012, pp. 217--243. Ed.: R.P.\ Mondaini. World Scientific, 2013.}

\maketitle

\abstracts{
In the present work we analyze the problem of adaptation and evolution of RNA virus
populations, by defining the basic stochastic model as a multivariate branching process 
in close relation with the branching process advanced by Demetrius, Schuster and Sigmund 
(``Polynucleotide evolution and branching processes'', Bull. Math. Biol. 46 (1985) 
239-262), in their study of polynucleotide evolution. We show that in the absence of
beneficial forces the model is exactly solvable. As a result it is possible to prove
several key results directly related to known typical properties of these systems like
(i) proof, in the context of the theory of branching processes, of the lethal
mutagenesis criterion proposed by Bull, Sanju\'an and Wilke (``Theory of lethal
mutagenesis for viruses'', J. Virology 18 (2007) 2930-2939); (ii)  a new proposal for
the notion of relaxation time with a quantitative prescription for its evaluation and
(iii) the quantitative description of the evolution of the expected values in four
distinct regimes: transient, ``stationary'' equilibrium, extinction threshold and lethal
mutagenesis. Moreover, new insights on the dynamics of evolving virus populations can be
foreseen.
}

\section{Introduction}
\label{sec:INTRO}

RNA viruses exhibit a pronounced genetic diversity\cite{D1985}. 
This variability allows RNA virus to better adapt to environmental challenges as
represented by host and therapy pressures\cite{D1998}.
Due to the lack of a proofreading activity of viral RNA polymerases
(average error incorporation rate in the order of $10^{-4}$ per nucleotide, per
replication cycle\cite{SDH92}), short generation times and huge population numbers,
RNA viral populations may be viewed as a collection of particles bearing mutant genomes. 
As a consequence of high mutation rates, frequencies of mutants depend not only on their
level of adaptation but on the probability of being faithfully replicated during viral
genomic RNA synthesis. Several studies have looked at viral diversification processes
as a contributing cause of disease progression and of therapeutic strategies
shortcomings including vaccine trials\cite{D1998,ES05}. 
It has become important to understand the process
by which virus acquire diversity and the dynamics and fluctuations of this diversity in
time. However, understanding viral evolution \emph{in vivo} has proven to be a very
cumbersome accomplishment due to the so many variables present in the interplay
between virus and their hosts.
To name a few; the host defense pressures as the innate and ``cognitive'' immune
responses, the use of antiviral drugs, the turnover rate of virus populations composed
by viral replication and clearance, the elevated mutational rates of RNA virus, and the
possible existence of structured viral reservoirs in infected patients. It is also
important to take into account the size of the viral innoculums at the moment of
infection, how frequent viral populations undergo bottleneck passages within a host, how
differently each infected individual may react to an incoming virus and finally how many
viral variants are tightly associated with differential biological capabilities.

Traditionally, in an effort to make the viral evolution process more palpable, several
groups have addressed this subject from different points of view.
There is a substantial amount of publications that studied virus populations during
their evolution in experimental settings, for instance, cell 
cultures\cite{ELM06}, by challenging the virus with population
bottlenecks\cite{OBD10,OBANED10}, or the introduction of antiviral drugs\cite{CCA01},
including mutagens, or another competing viral population.
Experiment outcomes were evaluated using viral replication kinetics, the
intensity and quality of the observed mutational spectra and virus survival/extinction
as final parameters. 
Several other groups have studied the process of viral evolution 
away from the bench but using mathematical and computational 
tools\cite{WWOLA01,LEDM02,MLPED03,BSW07,ALM09}.
These models are quite tractable but there is always the risk of oversimplification.
To escape from oversimplifying the interplay between virus and hosts a model needs to
incorporate a few hard rules based on previous experimental data which has been
generated by the whole community of investigators addressing viral evolution. 
Based on other groups experimental data and previous mathematical models put forward by
other investigators as the one presented by L\'azaro \etal\cite{LEDM02} we sought to
study a stochastic model for virus evolution that would be able to describe some
general aspects of RNA virus evolution.
Here, RNA viral evolution is described by a multivariate branching process during which 
each round of replication is accompanied by the introduction of a single point mutation 
per genome in the viral progeny.
More explicitly, by a re-interpretation of total mutation probabilities of a genome
as probabilities of ``replicative effects'' we obtain a non-trivial multivariate
generalization of the single-type branching process of Demetrius \etal\cite{DSS85}.

Drake and Holland\cite{DH99} back in 1999 have inferred, based on limited data, a
central value for the RNA virus mutation rate per genome per replication of
$\mu_\mathrm{g} \approx 0.76$ and suggested the rate per round of cell infection of
$\mu_\mathrm{g} \approx 1.5$. 
In 2010, Sanju\'an \etal\cite{SNCMB10} revisiting this theme by reviewing a list of
previous publications encountered RNA virus mutational rates in the order of $10^{-4}$
to $10^{-6}$ with $\mu_\mathrm{g}\approx 4.64$ for the bacteriophage $\mathrm{Q\beta}$ 
(Batschelet \etal\cite{BDW76}) and $\mu_\mathrm{g} \approx 1.15$ for hepatitis C virus 
(Cuevas \etal\cite{CCMS09}).

It has been demonstrated that virus populations may be reduced at the moment of
infection, and only a few particles are able to start a new infection process in naive
hosts\cite{ZDE11,K08}.
Abrupt reductions on RNA viral populations known as population bottlenecks may eliminate
population diversity and lead the virus to pathways towards extinction due to the
exacerbated effects of genetic drift. An incoming virus population recovering from a
transmission bottleneck event may show an asymptotic behavior resembling stationary
equilibrium represented by the balance between two opposite forces classically identified
with mutation and selection. 
This asymptotic behavior would occur if the environment is constant and enough
time is allowed between two successive bottleneck events.
The relaxation time between the bottleneck and the establishment of stationary
equilibrium has been referred to as the ``recovering time'' by Aguirre \etal\cite{ALM09}.

As pointed out by Drake and Holland\cite{DH99}, the basal value of RNA
virus mutation rates is so large and RNA virus genomes are so informationally dense,
that even a modest rate increase extinguishes the population. The frequent appearance of
overlapping reading frames and multifunctional proteins augments the risk of a random
mutation to have a deleterious impact and even more, multiply the effect of deleterious
mutations. For example, the fraction of deleterious mutations out of random mutations
occurring in vesicular stomatitis virus is around $70\%$ (Sanju\'an \etal\cite{SME04}).
If the introduction of a mutagen to a replicating virus population is able to cause its
extinction by increasing mutational rates, the process is known as chemical lethal
mutagenesis and has been demonstrated in a number of viruses including the
vesicular stomatitis virus (VSV)\cite{HDTS90,LGNHDH97}, 
human immunodeficiency virus type 1 (HIV-1)\cite{LEKZRM99},
poliovirus type 1\cite{CCA01,HDTS90}, 
foot-and-mouth disease virus\cite{SDLD00}, 
lymphocytic choriomeningitis virus\cite{GSCDL02}, 
Hanta virus\cite{SSJJ03} and 
Hepatitis C virus\cite{ZLBMR03}. 
Accordingly, in the model, increases on mutational rates, and
more specifically, on the deleterious component of the mutational spectrum are able to
push viral populations towards extinction. Our results corroborate with the study
from Bull, Sanju\'an and Wilke\cite{BSW07} by showing that the sufficient condition for
lethal mutagenesis involves mutational and ecological aspects.
Bull \etal\cite{BSW07} arrived at a conjectural criteria for lethal mutagenesis by
a heuristic and intuitive approach of possible general applicability.
By applying the branching process theory to the evolution of RNA viruses the lethal
mutagenesis inequality proposed by Bull \etal\cite{BSW07} is rigorously proven here. 
Furthermore, we describe four distinct regimes of RNA virus populations: transient
regime, stationary equilibrium, extinction threshold, and extinction through lethal
mutagenesis.

We note that in previous works\cite{ALM09,CCMA11,CACM11,MLPED03} the properties of
phenotypic models are discussed starting from a mean field linear model described by a
mean matrix without reference to any underlying stochastic process modeling the
microscopic dynamics of particle replication.
In fact, almost any stochastic model of asexual replication has an underlying branching
process, which remains implicit and undefined in most of the studies in this subject. 
However, there are a few of them\cite{DSS85,HRWB02,W03} that take a different path
and explicitly define the stochastic process in order to bring the mathematical theory
of branching processes to bear.
This attitude has some virtues, since it provides powerful tools, that have been
perfected in the past several decades, allowing one to extract quantitative results
on a rigorous basis in clear conceptual framework.
Therefore, we shall follow this path, by deducing the matrix of first moments from a
generating function of an explicitly defined stochastic process and explore the
mathematical consequences of such conceptual framework.

\ignore{
Although the two matrices happen to coincide, it is important to stress that
only from the generating function of the underlying stochastic process it is possible 
to fully discuss the validity of the model.
}

\pagebreak

\paragraph*{Structure of the Paper.}
In section~\ref{sec:PMVE} we describe a class of models for viral evolution and
show that they define a multitype branching process. We explicitly compute the
generating function and derive some elementary properties.
In section~\ref{sec:TSPM} we solve the spectral problem for the mean matrix of
the model which allows us to apply the theory of of multitype branching processes
in order to obtain our results.
Finally, in section~\ref{sec:CO} we present our conclusions and directions for
future research.
For the convenience of the reader, in the~\ref{sec:BRTBP} we briefly summarize all 
the definitions and results from the theory of multitype branching processes
that are used in our proofs.

\section{Phenotypic Models for Viral Evolution}
\label{sec:PMVE}

In this section we describe a model for viral evolution that is naturally represented
by a multivariate branching stochastic process generalizing, in a non-trivial way, the 
single-type branching process studied by Demetrius \etal\cite{DSS85}.
For the sake of motivation we start by recalling a probabilistic model introduced 
by L\'azro \etal\cite{LEDM02}.
We interpret the notion of mutation probability as a general effect of
probabilistic nature acting on the replication capability of individual viral particles,
considered here as a measure of the particle's fitness characterizing its phenotype.
This effect is summarized by the definition of a stationary probability 
distribution which is used to set up a Galton-Watson branching process (Watson and 
Galton\cite{WG1874}) for the temporal evolution of the viral population. 
This probability distribution gives appropriate parameters to classify the asymptotic
behavior of the viral population and to describe some of the non-equilibrium properties
of the model. 

In other related publications the concept of mutation is extensively used as the cause
of replication capacity change.
Understanding that those changes constitute an observable output due to many different
factors (of genetic and non-genetic nature), we prefer to use the general term
``effect'' over the replication capacity to characterize the three possible changes 
(deleterious, beneficial and neutral) that may happen with the viral particle when it 
replicates.
The precise definition of the three types of changes are given in the next section. 

\subsection{Definition of the Model}

A number of viral infections starts with the transmission of a relatively small number
of viral particles from one host organism to another one. The initial viral population starts
replicating constrained by the unavoidable interaction with the host organism and
evolves in time towards an eventual equilibrium. Each particle composing the population
replicates in the cellular context that may differ from cell to cell. Moreover each
particle has different replication capabilities due to the natural genomic diversity
found in viral populations in general. Therefore, it is reasonable to consider the viral
population as a set of particles divided in groups of different replication capabilities
measured in terms of the number of particles that one particle can produce. Each of
those groups we call a class; the replication capability of a viral particle is an
output of the process of interaction of that particle carrying its genetic information
with the cell environment. The replication capability is considered as a phenotypic
character of the particle and therefore each class is considered as a set of particles
with a possible genotype diversity expressing the same phenotype trait. The model we
consider here does not take into account any information about the genomic diversity of
any replicating class and therefore it should be classified as phenotypic model. 

We consider that the whole set of particles composing the viral population replicates at
the same time in such a way that the evolution of the population is described as a
succession of discrete viral generations. This assumption crucially depends on the clear
definition of the time needed for a particle to replicate, referred by virologists as
\emph{generation time}. As it depends on the cell environment it is clear that this time
period may vary from particle to particle replicating in different cells in such a way
that the meaningful concept is a distribution of replication times with a possible clear
mean value. The dispersion of the replication times can be considered small if we
restrict ourselves to homogeneous cell populations.
Under these conditions we consider that no particle can be part of two successive
generations. The possible impact of a subset of non replicating particles on the
dynamics of the viral population is left to further studies. 

Suppose that we have a population of viruses that start evolving from an initial
set of particles (population at $t=0$), which is partitioned into \emph{classes}
according to the \emph{replication capacity} of each particle, that is, where each
particle of class $0$ produces no copies of itself, each particle with class $1$
produces one copy of itself, and so on.
We assume that there is a \emph{maximum replication capacity} $R$ imposed by the natural
limiting conditions under which any particle of the population replicates. Moreover, as
the process of replication is controlled by chemical reactions involving specific
enzymes and the template, it is reasonable to assume a mean bounded replication capacity
per particle that is possibly typical for each specific virus.

\enlargethispage{5mm}

In the process of replication of a viral particle errors may occur at each
replication cycle in the form of point mutations with possible impact on the
replication capacity of the progeny particles. Due to the intrinsic stochastic
component of chemical reactions it is natural to treat this point mutational cause as
probabilistic. Another possible cause of change in the replication capability in the
viral offspring is clearly related to the cellular environment where the replication
process takes place. As a result the time evolution of viral populations should be
viewed as a physical process strongly influenced by stochasticity. Therefore we consider
that the combined action of genetic and non-genetic causes may produce basically three
types of replicative effects namely:
\begin{itemize}
\item \emph{deleterious effect}: the replication capacity of the copied particle
      decreases by one. When the particle has capacity of replication equal to $0$
      it will not produce any copy of itself.
\item \emph{beneficial effect}: the replication capacity of the copy increases by one.
      If the replication capacity is already the maximum allowed then the replication
      capacity of the copies will stay the same.
\item \emph{neutral effect}: the replication capacity of the copies is
      the same as the replication capacity of the parental particle.
\end{itemize}
For each type of effect we associate a probability at the particle scale applicable to
every single replication event: $d$ for the probability of the occurrence of a
\emph{deleterious} effect, $b$ for the probability of the occurrence of a
\emph{beneficial} effect and the complementary probability $c=1-b-d$ is the probability
of the occurrence of a neutral effect.
In the case of \emph{in vitro} experiments with homogeneous cell populations the
parameters $c$, $d$ and $b$ may be considered as mutation probabilities.

The \emph{simple phenotypic model} is obtained by requiring that there are no
beneficial effects in time, that is $b=0$.
This assumption is justified by several experimental results. 
The frequencies between beneficial, deleterious and neutral mutations appearing in a
replicating population have been already measured by prior 
studies\cite{MGME99,IS01,O03,SME04,CIE07,EWK07,PFCM07,RBJFBW08}.
Taking their results together, it is reasonable to conclude that beneficial mutations
could be as low as 1000 less frequent than either neutral or deleterious mutations. 
As a result the viral population would be submitted to a large number of successive
deleterious and neutral changes and a comparatively small number of beneficial changes.
Here we shall focus on the case $b=0$ which allows us to obtain exact results.
The case $b \approx 0$ can be treated perturbatively as discussed in Antoneli, Bosco,
Castro and Janini\cite{ABCJ12}.

From what is described above it should become clear that the model assumes a scenario
where a probabilistic processes at the cellular/viral scale take place in
the context of the interaction between the viral particle and the host cell.
The combined effect of small scale processes are observed at the viral population
scale in terms of collective (stable or not) properties.

Based on the general aspects of the phenomenon of viral replication it is compelling to
to model it in terms of a branching process. 
In this perspective we define a \emph{discrete multitype Galton-Watson branching process}
for the evolution of the initial population, where the \emph{classes}
will be represented by the replication capabilities $0,1,\ldots,R$. 
The branching process is described by a sequence of vector-valued random variables
$\{\bm{Z}_n:n\in\N\}$ giving the number of particles in each replication class 
in the $n$-th generation.
Thus $\bm{Z}_n$ are vectors of non-negative integers satisfying the following
assumption: if the size of the $n$-th generation is known, then the probability laws
governing the later generations does not depend on the sizes of generations preceding
the $n$-th, that is the sequence $\{\bm{Z}_n:n\in\N\}$ forms a \emph{markovian process}.
The initial population $\bm{Z}_0$ is represented by a vector of non-negative integers
$\bm{Z}_0=(Z_0^0,Z_0^1,\ldots,Z_0^R)$, which is non-zero and non-random.
The temporal evolution of the population is obtained from a vector-valued discrete
probability distribution $\bm{\zeta}=(\zeta_0,\zeta_1,\ldots,\zeta_R)$ defined on the
set of vectors with non-negative integer entries called the 
\emph{offspring distribution} of the branching process. 
For any vector with non-negative entries $\bm{i}=(i^0,\ldots,i^R)$ one has that
\begin{equation} \label{eq:BPLAW}
 \Prob(\bm{Z}_{n+1}=\bm{i}|\bm{Z}_n=\bm{e}_r)=\zeta_r(\bm{i}) \,,
\end{equation}
where $\bm{e}_r=(0,\ldots,1,\ldots,0)$, with $1$ in the $r$-th position.
Thus, $\zeta_r(\bm{i})$ is the joint probability that an individual particle of 
class $r$ ($0\leqslant r\leqslant R$) generates $i^0$ progeny particles in the class
$0$, $i^1$ progeny particles in the class $1$, \ldots, $i^R$ progeny particles in the
class $R$.
Note that any vector $\bm{Z}_n=(Z_n^0,Z_n^1,\ldots,Z_n^R)$ may be written as a sum
$\sum_r Z_n^r\bm{e}_r$ and since each particle in $\bm{Z}_n$ may be seen as the
initial condition of a new branching process independently of the others, 
equation~(\ref{eq:BPLAW}) determines the probability laws for a general branching 
process as follows
\[
 \Prob\big(\bm{Z}_{n+1}=\bm{i}|\bm{Z}_n={\textstyle\sum}_r Z_n^r\bm{e}_r\big)
 =\prod_r \zeta_r(\bm{i})^{Z_n^r} \,.
\]

In order to compute the offspring probability distribution $\bm{\zeta}$ for the 
simple phenotypic model, we start by observing that $\zeta_r$ is non-zero only when
$\bm{i}$ is of the form $\bm{i}=(0,\ldots,i^{r-1},i^{r},\ldots,0)$ since a particle with
replication capability $r$ can only produce progeny particles of the replication
capability $r$ or $r-1$, moreover the entries $i^{r-1}$ and $i^{r}$ should satisfy
$i^{r-1}+i^{r}=r$.
Thus we just need to compute the probabilities $\zeta_r$ on the vectors of the form
$\bm{i}_k=(0,\ldots,r-k,k,\ldots,0)$.
Suppose that a viral particle $v$ with replication capacity $r$
($0 \leqslant r \leqslant R$) replicates itself producing new virus particles
$v_1,\ldots,v_r$.
For each new particle $v_i$, there are two possible outcomes regarding the type of
change that may occur: neutral or deleterious, with probabilities $c=1-d$ and $d$,
respectively. 
Representing the result of the $i$-th replication event by a variable $X_i$, 
which can assume two values: $0$ if the effect is deleterious (failure) and $1$ 
if the effect is neutral (success), the probability distribution of $X_i$ is that
of a \emph{Bernoulli trial} with probability of occurrence of a neutral effect $c=1-d$
(success), that is,
\[
 \Prob(X_i=k)=(1-d)^k \, d^{1-k} \qquad (k=0,1) \,.
\]
The total number of neutral effects that occur when the original virus particle
reproduces is a random variable $S_r$ given by the sum of all variable $X_i$,
since each copy is produced independently of the others,
\[
 S_r=X_1+X_2+\ldots+X_r \,.
\]
That is, $S_r$ counts the total number of neutral effects (successes)
that occurred in the production of $r$ virus particles $v_1,\ldots,v_r$. 
It also represents the total number of particles that will have the same
replication capacity $r$ of the original particle $v$.
It is well known (see Feller\cite{F68}) that a sum of $r$ independent and identically
distributed Bernoulli random variables with probability $c=1-d$ of success has a
probability distribution given by the \emph{binomial distribution}:
\[
 \Prob(S_r=k)=\mathrm{binom}(k;r,1-d)={r \choose k} \, (1-d)^k \, d^{r-k} \,.
\]
Since this is the probability that a class $r$ virus particle $v$ produces $k$
progeny particles with the same replication capability as itself one has
therefore
\[
 \zeta_r(0,\ldots,r-k,k,\ldots,0)=\Prob(S_r=k)=\mathrm{binom}(k;r,1-d) \,.
\]

Given the offspring probability distribution $\bm{\zeta}$ one may set up a
\emph{probability generating function} $\bm{f}=(f_0,\ldots,f_R)$ which is
defined by the power series
\[
 f_r(z_0,z_1,\ldots,z_R)=\sum_{\bm{i}} \,\zeta_r(\bm{i}) \, z_0^{i^0}\ldots z_R^{i^R} \,.
\]
The probability generating function of the simple phenotypic model is
\begin{equation} \label{EQ:genfunc1}
\begin{split}
 f_0(z_0,z_1,\ldots,z_R) & = 1 \\
 f_1(z_0,z_1,\ldots,z_R) & = dz_0+cz_1 \\
 f_2(z_0,z_1,\ldots,z_R) & = (dz_1+cz_2)^2 \\[-2mm]
                         & \;\;\vdots \\[-2mm]
 f_R(z_0,z_1,\ldots,z_R) & = (dz_{R-1}+cz_R)^R
\end{split}
\end{equation}
Note that the functions $f_r$ are polynomials whose coefficients are exactly
$\mathrm{binom}(k;r,1-d)$.
This function completely determines the branching process.

Now it is easy to obtain the general case where the beneficial effects
have a non-zero contribution $b$. 
In this case, the binomial distribution is replaced by a 
\emph{trinomial distribution} (see Feller\cite{F68}) and the probability
generating function of the general phenotypic model is 
\begin{equation} \label{EQ:genfunc2}
\begin{split}
 f_0(z_0,z_1,\ldots,z_R) & = 1 \\
 f_1(z_0,z_1,\ldots,z_R) & = dz_0+cz_1+bz_2 \\
 f_2(z_0,z_1,\ldots,z_R) & = (dz_1+cz_2+bz_3)^2 \\[-2mm]
                         & \;\;\vdots \\[-2mm]
 f_{R-1}(z_0,z_1,\ldots,z_R) & = (dz_{R-2}+cz_{R-1}+bz_R)^{R-1} \\
 f_R(z_0,z_1,\ldots,z_R) & = (dz_{R-1}+(c+b)z_R)^R
\end{split}
\end{equation}

\begin{remark} \sl \label{rmk:VAR1}
It is worth to mention that there are other variations of these models that share
the same essential properties and are more adequate in different contexts.
\begin{itemize}
\item \textbf{With Zero Class:} In this variation, which is the version deduced above, 
      particles of class $r=0$ are generated by the particles from class $r=1$. 
\item \textbf{Without Zero Class:} In this variation, the particle class $0$ is omitted 
      and thus the probability generating function has $R$ variables and $R$
      components: omit the variable $z_0$, the first component $f_0$ and define
      $f_1(z_1,\ldots,z_R) = d+cz_1+bz_2$.
      Particles of class $r=1$ undergoing a deleterious change are eliminated in the
      next generation.
\end{itemize}
\end{remark}

\subsection{Basic Properties of the Phenotypic Model}

We start by recalling that, when calculating probabilities and expectations,
there is no loss of generality if one considers only initial populations consisting of
just one particle of class $r$ ($0 \leqslant r \leqslant R$), since the general case can
be decomposed as a sum of independent processes with this kind of initial population.
All the relevant properties of the model can be deduced with this simplification. 

We shall introduce the notation $Z_0^r=1$ for the condition $\bm{Z}_0=\bm{e}_r$,
which is the initial population consisting of one particle of class $r$ and zero
particles of other classes.
Thus $\Prob(\bm{Z}_1=\bm{i}|Z^r_0=1)=\zeta_r(\bm{i})$.
A basic assumption in the theory of branching processes is that all the first
moments are finite and that they are not all zero. 
Then one may consider the \emph{mean evolution matrix} or the
\emph{matrix of first moments} $\bm{M}=\{M_{ij}\}$ which describes how the averages
of the sub-populations of particles in each replication class evolves in time:
\[
 M_{ij}=\Expec(Z_1^i|Z_0^j=1) \,,\qquad\forall\, i,j=0,\ldots,R \,.
\]
In terms of the probability generating function one has
\[
 M_{ij}=\dfrac{\partial f_j}{\partial z_i}(1,1,\ldots,1) \,.
\]
Denoting by $\bm{f}'$ the jacobian matrix of $\bm{f}$ one may write
\[
 \bm{M}=\bm{f}'(\bm{1}) \,,
 \quad\text{where}\quad\bm{1}=(1,1,\ldots,1) \,.
\]
The evolution of the averages $\langle \bm{Z}_{n} \rangle$ of $\bm{Z}_n$ 
is given by
\begin{equation} \label{eq:MEANEVOL1}
 \langle \bm{Z}_{n} \rangle=\Expec(\bm{Z}_n|\bm{Z}_0)=\bm{M}^n\,\bm{Z}_0 \,.
\end{equation}

From the generating functions (\ref{EQ:genfunc1}) and (\ref{EQ:genfunc2}) 
it is trivial to compute the mean matrix of the phenotypic model. 
The mean matrix of the simple phenotypic model is
\begin{equation} \label{eq:MEANBASIC}
\bm{M}=\begin{pmatrix}
 0 & d &  0 &  0 &  0 & \cdots & 0 \\
 0 & c & 2d &  0 &  0 & \cdots & 0 \\
 0 & 0 & 2c & 3d &  0 & \cdots & 0 \\
 0 & 0 &  0 & 3c & 4d & \cdots & 0 \\
 0 & 0 &  0 &  0 & 4c & \cdots & 0 \\[-2mm]
\vdots & \vdots & \vdots & \vdots & \vdots & \ddots & Rd \\
 0 & 0 &  0 & 0 & 0 & 0 & Rc
\end{pmatrix}
\end{equation}
Note that it is an upper triangular matrix.

In the case of the general phenotypic model the mean matrix is
\begin{equation} \label{eq:MEANGENERAL}
\bm{M}=\begin{pmatrix}
 0 & d &  0 &  0 &  0 & \cdots & 0 \\
 0 & c & 2d &  0 &  0 & \cdots & 0 \\
 0 & b & 2c & 3d &  0 & \cdots & 0 \\
 0 & 0 & 2b & 3c & 4d & \cdots & 0 \\
 0 & 0 &  0 & 3b & 4c & \cdots & 0 \\[-2mm]
\vdots & \vdots & \vdots & \vdots & \vdots & \ddots & Rd \\
 0 & 0 &  0 & 0 & 0 & \cdots & R(c+b)
\end{pmatrix}~.
\end{equation}
Interestingly, the mean matrix $\bm{M}$ provided by this class of models is a
\emph{tridiagonal matrix}, which is an ubiquitous type of matrix appearing in 
several fields ranging from statistical signal processing\cite{G06},
information theory\cite{GS58}, lattice dynamical systems\cite{ADGW05}.

\begin{remark} \sl
In Demetrius \etal\cite{DSS85} a single type branching process is proposed as a model 
for the evolution of polynucleotides, moreover they state that state that the results
they obtained for the single type model are still valid in a multitype situation,
provided one excludes the possibility of ``back mutations'', i.e., if for every
replicative class $r$ the mutation rates from the classes $s$ to $r$, with
$s \leqslant r$, can be neglected.
This implies, in particular, that the mean matrix is \emph{upper triangular}. 
In fact, it is easy to see that, in the simplest case where the mutation rates
are non-zero only for adjacent classes, it is formally identical to the mean matrix
(\ref{eq:MEANBASIC}).
Moreover, one of the main assumptions of the basic theory of multitype branching
processes is that the mean matrix is \emph{primitive}\cite{H63,AN72}, that is, some
power of the matrix is positive.
However, an upper triangular mean matrix is never primitive and hence the underlying
branching processes, in principle, are out of the reach of the theory.
Fortunately, as it will be shown here, there is a generalization of the theory of
multitype branching processes which is capable of embracing the simple phenotypic model.
\end{remark}

The mean matrix of the phenotypic model can be viewed as the adjacency matrix
of a directed weighted graph where the nodes represent the particle classes
according to their replication capacity and the arrows represent the effect of
decrease or increase of the replication capacity due to the replication process
(see Figure~\ref{fig:GRAPH}).

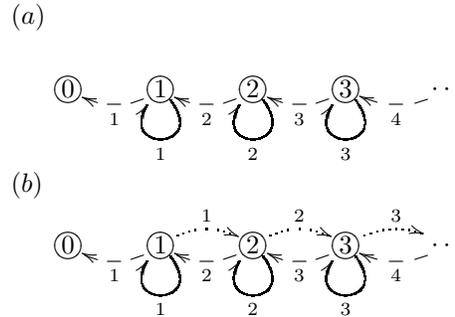
\begin{figure}[!htb]
\begin{center}
\[
\begin{array}{lc}
(a) \\[5mm]
& \xymatrix{
 *+[o][F]{0} & 
 *+[o][F]{1} \ar@/^/@{-->}[l]^1 \ar@(dr,dl)[]^1 & 
 *+[o][F]{2} \ar@/^/@{-->}[l]^2 \ar@(dr,dl)[]^2 & 
 *+[o][F]{3} \ar@/^/@{-->}[l]^3 \ar@(dr,dl)[]^3 & 
\cdots \ar@/^/@{-->}[l]^4
} \\
(b) \\
& \xymatrix{
 *+[o][F]{0} & 
 *+[o][F]{1} \ar@/^/@{-->}[l]^1 \ar@/^/@{.>}[r]^1 \ar@(dr,dl)[]^1 & 
 *+[o][F]{2} \ar@/^/@{-->}[l]^2 \ar@/^/@{.>}[r]^2 \ar@(dr,dl)[]^2 & 
 *+[o][F]{3} \ar@/^/@{-->}[l]^3 \ar@/^/@{.>}[r]^3 \ar@(dr,dl)[]^3 & 
\cdots \ar@/^/@{-->}[l]^4
}
\end{array}
\]
\caption{\label{fig:GRAPH} Graphs of mean matrices. (a) Simple phenotypic model.
         (b) General phenotypic model. The arrows are numbered according to which there
         occurs a deleterious effect ($d$ -- dashed arrows) or a beneficial effect
         ($b$ -- dotted arrows) or neutral effect ($c$ -- solid arrows).
}
\end{center}
\end{figure}

\section{The Simple Phenotypic Model}
\label{sec:TSPM}

For the simple phenotypic model it is easy to compute the eigenvalues
$\lambda_r$ of the mean matrix $\bm{M}$: 
\[
 \lambda_r=rc=r(1-d) \qquad r=0,\ldots,R \,.
\]
In particular, the malthusian parameter is the largest positive eigenvalue
\begin{equation} \label{eq:MALTHPAR}
 m=\varrho(\bm{M})=\lambda_R=Rc=R(1-d) \,.
\end{equation}
Therefore we have the following immediate result.
\begin{theorem} \label{thm:BASIC}
The simple phenotypic model has three distinct regimes.
\begin{enumerate}[(i)]
\item If $R(1-d)<1$ then the branching process is \emph{sub-critical}.
      That is, with probability $1$, the virus population becomes extinct
      in finite time.
\item If $R(1-d)>1$ then the branching process is \emph{super-critical}.
      That is, with positive probability, the virus population survives
      and grows indefinitely at an exponential rate proportional to $m^n$
      when $n\to\infty$.
\item If $R(1-d)=1$ then the branching process is \emph{critical}.
      That is, with probability $1$, the virus population becomes
      extinct but this may take an infinite time to happen. 
\end{enumerate}
\end{theorem}
\begin{proof}
This is a straightforward consequence of the classification of multitype branching
processes, as generalized by Sevastyanov\cite{H63,J70}, which is necessary to include
the simple phenotypic model, and equation (\ref{eq:MALTHPAR}).
\end{proof}

Theorem~\ref{thm:BASIC} provides a partition of the parameter space of the simple 
phenotypic model $\{(d,R):d\in [0,1],\,R\in\N\}$ into two regions
(see Figure~\ref{fig:BASIC1}).
The \emph{survival region} defined by $R>1/(1-d)$ and the \emph{extinction region}
defined by $R<1/(1-d)$.
The curve $R=1/(1-d)$ gives the \emph{extinction threshold}.

\begin{figure}[!hbt] 
\begin{center}
 \includegraphics[scale=0.4,angle=0]{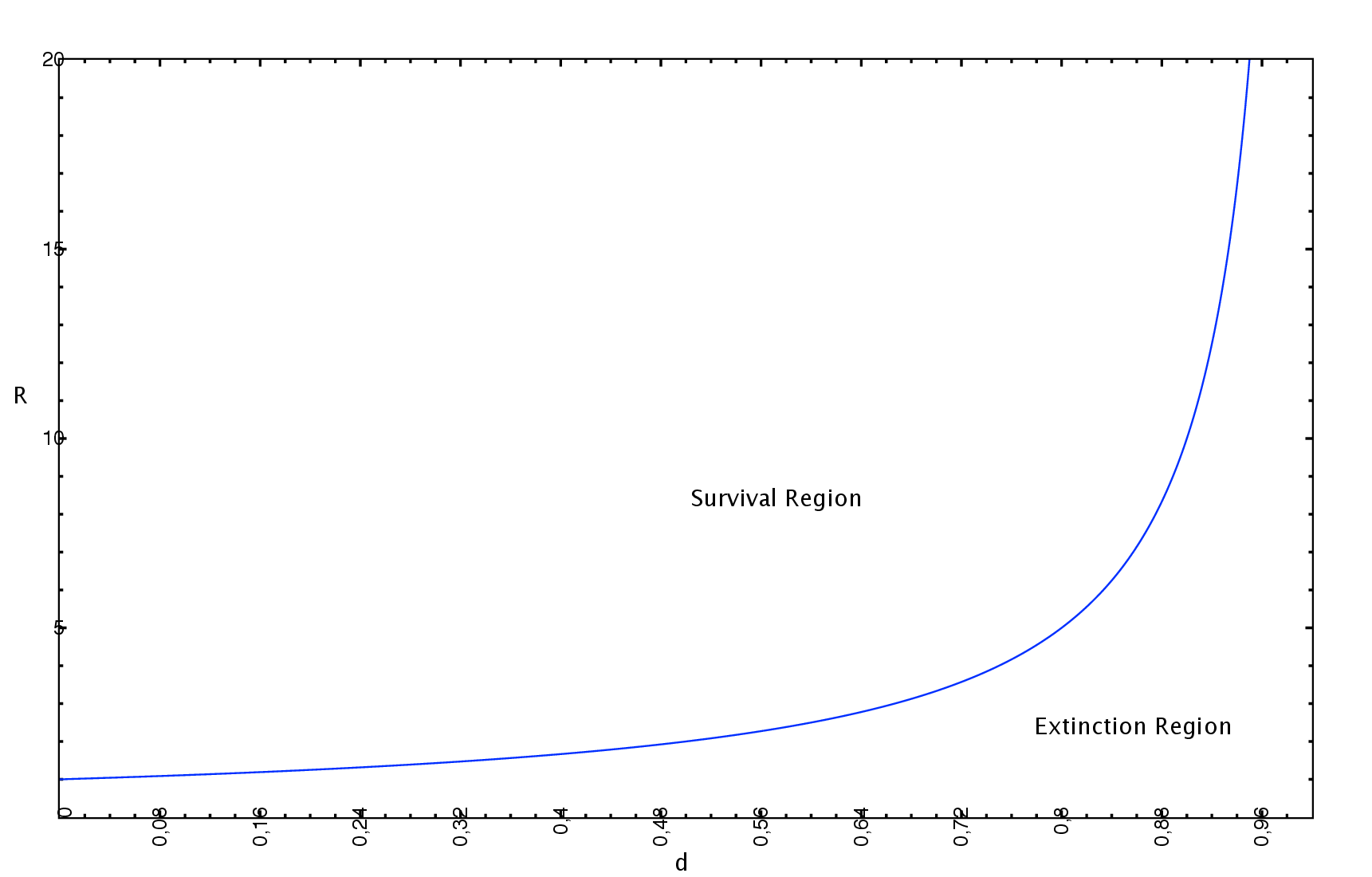}
 \caption{\label{fig:BASIC1} Graph of the function $R=1/(1-d)$ (in blue).
          The region below this curve corresponds to the sub-critical
          parameters $(d,R)$ and the region above this curve corresponds to the
          super-critical parameters $(d,R)$.
          The curve itself corresponds to the critical parameters $(d,R)$.}
\end{center}
\end{figure}

It is also important, specially in order to describe the asymptotic behaviour
in the super-critical case, to know the left eigenvectors $\bm{v}$
and right eigenvectors $\bm{u}$ corresponding to the eigenvalue $\lambda_R$
Let us write the left and right eigenvectors in components as
\[
 \bm{v}=(v_0,v_1,\ldots,v_R)
 \quad\text{and}\quad
 \bm{u}=(u_0,u_1,\ldots,u_R)
\]
and assume that they are normalized in the following way:
\[
 \bm{v}^{\mathrm{t}}\bm{u}=1 
 \quad\text{and}\quad
 \bm{1}^{\mathrm{t}}\bm{u}=1 \,.
\]
Then we have the following.
\begin{enumerate}[(i)]
\item In the version ``with zero class'' the left eigenvector $\bm{v}$ is given by
      \[
       \bm{v}=\dfrac{1}{(1-d)^{R}} \, (0,\ldots,0,1)
      \]
      and the right eigenvector $\bm{u}$ have coordinates $u_k$ given by
      \[
       u_k={R \choose k} \, (1-d)^k \, d^{R-k}=\mathrm{binom}(k;R,1-d) \,.
      \]
\item In the version ``without zero class'' there is no components $v_0$ and $u_0$.
      The left eigenvector $\bm{v}$ is given by
      \[
       \bm{v}=\dfrac{1-d^R}{(1-d)^{R}} \, (0,\ldots,0,1)
      \]
      and the right eigenvector $\bm{u}$ have coordinates $u_k$ given by
      \[
       u_k=\dfrac{1}{1-d^R} \, {R \choose k} \, (1-d)^k \, d^{R-k} \,.
      \]
\end{enumerate}
It is interesting to note that the simple phenotypic model is a
``completely solvable'' branching process in the sense that we
may explicitly solve the spectral problem for its mean matrix 
independently of the numerical values of the parameters.

Next we turn to the computation of the extinction probabilities $\gamma_r$.
In this case, it is necessary to solve a non-linear system of polynomial equations:
\begin{equation} \label{eq:SYSEQ}
\begin{split}
 z_0 & = 1 \\
 z_1 & = dz_0+(1-d)z_1 \\
 z_2 & = \big(dz_1+(1-d)z_2\big)^2 \\[-3mm]
     & \;\;\vdots \\[-3mm]
 z_R & = \big(dz_{R-1}+(1-d)z_R\big)^R
\end{split}
\end{equation}
This may be done in a recursive way, since the equation for $z_0$ is already
solved $z_0=1$ and the equation for $z_k$ depends only on $z_k$ and $z_{k-1}$.
Thus we get for $R=0,1,2$:
\[
\begin{split}
 \gamma_0 & = 1 \\
 \gamma_1 & = 1 \\
 \gamma_2 & = \left\{\begin{array}{l@{\quad\text{for}\quad}l}
          d^2/(1-d)^2 & 0\leqslant d \leqslant \tfrac{1}{2} \\[2mm]
          1 & \tfrac{1}{2}\leqslant d \leqslant 1
         \end{array}\right.
\end{split}
\]
When $R\geqslant 3$ the formulas become very complicated and when $R\geqslant 5$
the equation may not even be solvable by radicals, but in general one may write
\[
  \gamma_r = \left\{\begin{array}{l@{\quad\text{for}\quad}l}
          f(d) & 0\leqslant d \leqslant d_c \\[2mm]
          1 & d_c \leqslant d \leqslant 1
         \end{array}\right.
\]
where $d_c=\tfrac{r-1}{r}$ and $f(d)$ is a strictly increasing smooth function on
$[0,1[$ satisfying: (i) $f(0)=0$, (ii) $f(d_c)=1$, (iii) $f(d)<1$ 
for $0\leqslant d < d_c$ and (iv) $\lim_{d\to 1}f(d)=+\infty$.
This expression suggests that the surviving probabilities $\omega_r=1-\gamma_r$
can be interpreted as an \emph{order parameter} associated to the occurrence of a
\emph{phase transition} when the deleterious probability $d$ attains the
critical point $d_c=\tfrac{r-1}{r}$, which marks the transition from
super-criticality to sub-criticality,
\[
  \omega_r = \left\{\begin{array}{l@{\quad\text{for}\quad}l}
          g(d) & 0\leqslant d \leqslant d_c \\[2mm]
          0 & d_c\leqslant d \leqslant 1
         \end{array}\right.
\]
where $g(d)=1-f(d)$ and thus satisfies: (i) $g(0)=1$, (ii) $g(d_c)=0$, 
(iii) $g(d)>0$ for $0\leqslant d < d_c$ and (iv) $\lim_{d\to 1}g(d)=-\infty$.
Observe that for a fixed numerical value of $d$, the system of equations
(\ref{eq:SYSEQ}) can be easily solved by numerical approximation using Newton's method.
For instance, in Figure~\ref{fig:BASIC2} we show the curves for the surviving
probability function $\omega_r(d)$.

\begin{figure}[!htb] 
\begin{center}
 \includegraphics[scale=0.19,angle=0]{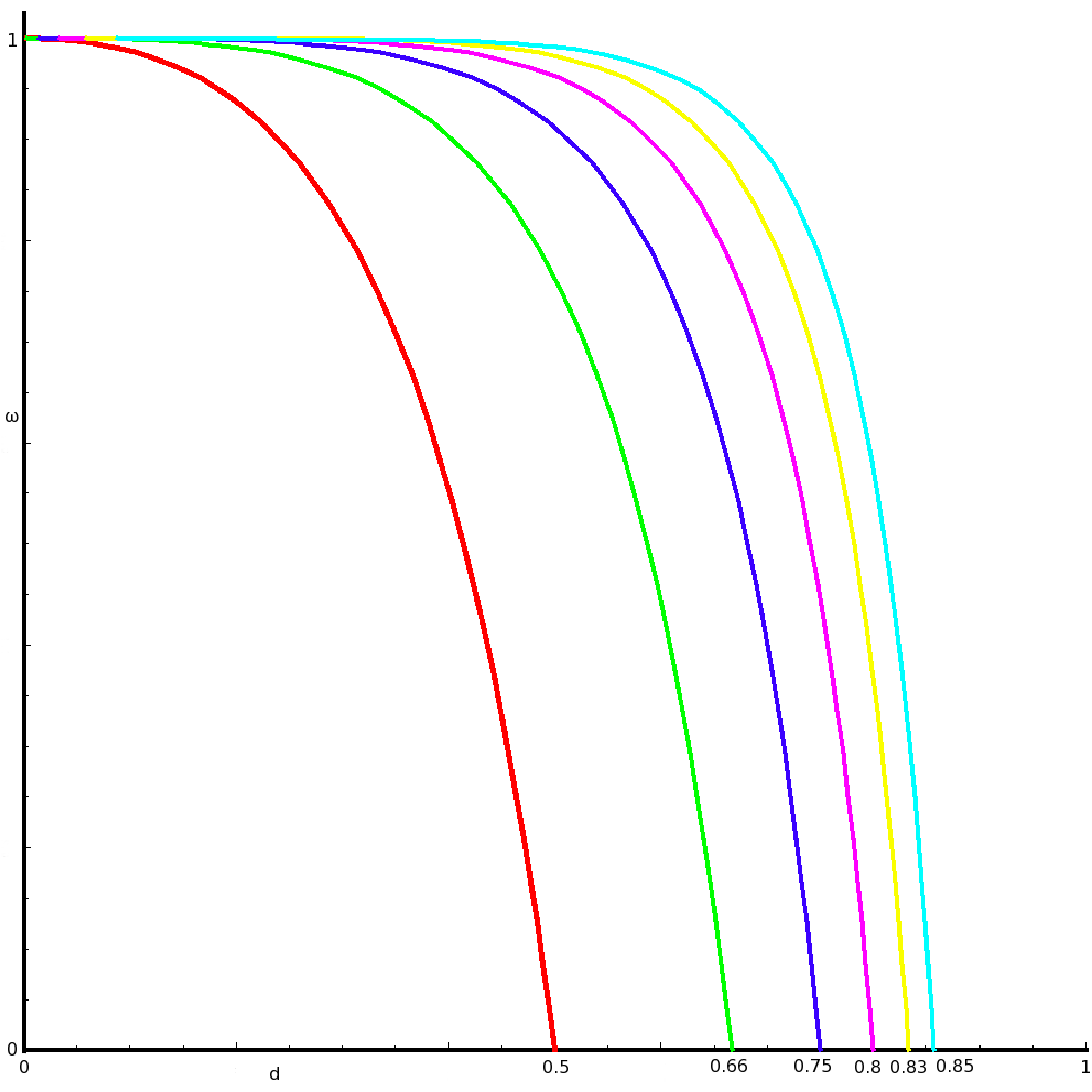}
 \caption{\label{fig:BASIC2} Curves for the surviving probability $\omega_r(d)$
          as function of $d$ for $r=2,\ldots,7$.}
\end{center}
\end{figure}

The result shows that, with respect to $\omega_r$, the model has a critical behavior
in complete analogy to a second order phase transition (see Figure~\ref{fig:BASIC2}).
Therefore, the critical properties of the model can be characterized by means of
relevant critical exponents.

Finally, it is not difficult to see that for fixed $d$, the numbers $\gamma_r$ satisfy
$1\geqslant\gamma_2\geqslant\gamma_3\geqslant\ldots\geqslant\gamma_R$
and therefore the extinction probability for a general initial condition 
$\bm{Z}_0=(Z_0^0,\ldots,Z_0^R)$ may be estimated far from the critical deleterious
probability $d_c=(R-1)/R$ by
\begin{equation} \label{eq:SURVPROB1}
 \Prob(\bm{Z}_n=0 \;\text{for some $n$}|\bm{Z}_0)~\approx~\gamma_2^{|\bm{Z}_0'|} \,,
\end{equation}
where $|\bm{Z}_0'|=Z_0^2+\ldots+Z_0^R$ and near $d_c=(R-1)/R$ by
\begin{equation} \label{eq:SURVPROB2}
 \Prob(\bm{Z}_n=0 \;\text{for some $n$}|\bm{Z}_0)~\approx~\gamma_R^{Z_0^R} \,.
\end{equation}

It has been demonstrated that large population passages are able to increase the
adaptability of virus populations\cite{LEDM02}. 
On the other hand, small population passages represented by bottleneck events are 
capable to increase the risk towards viral extinction. 
Among the aspects of abrupt population reductions are the exacerbated
effects of drift that coupled with the Muller's hatchet principle\cite{M64} 
may lead to the random and progressive lost of the best adapted virus in a population. 
It also has been suggested that large virus populations bearing a significant phenotypic
diversity are more adaptable to environment fluctuations and robust. It is correct to
assume that large initial virus populations colonizing new hosts may show better
survival probabilities than populations recovering from bottlenecks. In this way the
size of the viral innoculums may have an impact in the survival rates of different virus
populations.
It is important to note that the existence of a clear cut between regimes of survival
and non survival populations by means of a critical state is directly related to the
problem of lethal mutagenesis for viral populations.

From now on we shall split the analysis of the simple phenotypic model
according to which it is sub-critical, super-critical or critical.

\subsection{The Sub-critical Regime: Lethal Mutagenesis}

The first consequence of our results is a generalization, in the context of the
phenotypic model (provided one assumes that all effects are of purely mutational nature),
of the conjecture of \emph{lethal mutagenesis} of Bull, Sanju\'an and Wilke\cite{BSW07}.
Recall that Bull \etal\cite{BSW07} assume that all mutations are either neutral or
deleterious and write the mutation rate $U=U_d+U_c$ where the component $U_c$ comprises
the purely neutral mutations and the component $U_d$ comprises the mutations with a 
deleterious fitness effect.
Let $R_{\mathrm{max}}$ denote the maximum reproductive capacity among all particles in
the viral population. 
The \emph{extinction criterion} proposed by Bull \etal\cite{BSW07} states that a
sufficient condition for extinction is 
\begin{equation} \label{eq:LETHALMUT}
 \ee^{-U_d}R_{\mathrm{max}} < 1 \,.
\end{equation}
According to Bull \etal\cite{BSW07}, the factor $\ee^{-U_d}$ is both the mean fitness
level and also the proportion of offspring with no non-neutral mutations.
In the absence of beneficial mutations the only type of non-neutral mutations are the
deleterious mutations and hence 
\begin{equation} \label{eq:FUNDREL}
 \ee^{-U_d}~=~c~=~1-d \,.
\end{equation}

More precisely, if $\langle\bm{Z}_n\rangle=(Z^0,\ldots,Z^R)$ denotes the mean population
at time $n$ then the offspring at time $n+1$ is given by 
$\langle\bm{Z}_{n+1}\rangle=\bm{M}\langle\bm{Z}_n\rangle$ and thus
\[
\begin{split}
\ee^{-U_d} 
& = \dfrac{cZ^1+2cZ^2+\ldots+RcZ^R}{dZ^1+cZ^1+2dZ^1+2cZ^2+\ldots+RdZ^R+RcZ^R} \\
& = \dfrac{c}{d+c} \, \dfrac{Z^1+2Z^2+\ldots+RZ^R}{Z^1+2Z^2+\ldots+RZ^R} \\  
& = c = 1-d \,.
\end{split}
\]

Since the maximum reproductive capacity among all particles in the viral population in
our model is given by the maximum number of replicative classes $R$, it follows that
$R_{\mathrm{max}}=R$.
Therefore, the extinction criterion (\ref{eq:LETHALMUT}) is equivalent, in the context
of the simple phenotypic model, to
\begin{equation} \label{eq:LETHALMUTALT}
 (1-d)R < 1 \,,
\end{equation}
which is exactly the condition for the model to be sub-critical.

\begin{corollary} \label{thm:CORLM}
In the simple phenotypic model, the virus population becomes extinct in finite time,
with probability $1$, if the product of the \emph{neutral effect probability}
$(1-d)$ with the \emph{maximum replication capacity} $R$ is strictly less than $1$.
\end{corollary}

In another work, Bull \etal\cite{BSW08} suggest a modification of the extinction
threshold eq. (\ref{eq:LETHALMUT}) that accounts for beneficial effects as long as
they do not couple the deleterious ones (see Antoneli \etal\cite{ABCJ12} for a more
general result).

The main conclusion here is that the existence of lethal mutagenesis depends on
``genetic components'' (mutational rates) and other additional deleterious effects
(host driven pressures intensifications), as well as on strict ``ecological
components'', namely, the maximum replication capacity of the particles in the
population and on the initial population size.
As a result the viral population may reach extinction by increasing the number of
deleterious mutations per replication cycle, by decreasing the value of $R$ in the
population or by a combination of the two mechanisms. 
The mutational strategy is the basis of treatments using mutagenic drugs\cite{CCA01}
that induce errors in the generation process of new viral particles reducing their
replication capacity. 
A straightforward consequence of extinction criterion eq. (\ref{eq:LETHALMUT})
or eq. (\ref{eq:LETHALMUTALT}) is that a single particle
showing the maximum replication capacity $R$ is able to rescue a viral population driven
to extinction by mutagenic drugs. 
If it is assumed that RNA virus populations correspond to a swarm of variants with
distinct replication capacities, for a therapy to become effective it is important that
it will eliminate the classes represented by particles with highest replication
capacities.    
As a conclusion the higher the replication capacity of the first particles infecting the
organism the larger should be the number of deleterious mutations (or effects) and
therefore the larger should be the drug concentration. This can be a clear limitation
for treatments based on mutagenic drugs.

\subsection{The Super-critical Case: Relaxation and Equilibrium}

In the super-critical regime, the population grows at a geometric pace indefinitely.
Nevertheless, there are two distinct phases that occur during this growth:
a transient phase (``relaxation''or ``recovery time'') and a dynamical stationary phase.

\subsubsection{Relaxation towards equilibrium.}

An important question concerning the adaptation process of a viral population to the
host environment is the typical time needed to achieve the equilibrium state. As the
equilibrium is characterized by constant mean replication capacity an obvious criteria
to measure the time to achieve equilibrium would be by the vanishing variation of this
variable as used in other studies (Aguirre \etal\cite{ALM09}). 
Nevertheless, this method is clearly subjected to the limitations of numerical accuracy
with evident disadvantages if one wants a sharp and universal criterion to differentiate
populations from the point of view of how fast a population can be typically stabilized
in a organism. 

\pagebreak

Viral populations are commonly submitted to transient regimes. 
As pointed out earlier the infection transmission process represents the passage of a
small number of particles from one organism to another in such a way that in this process
the viral population is submitted to a subsequent \emph{bottle-neck effect} during
spreading of viruses in the host population. 
In order to approach the problem of relaxation after a bottleneck process in a more sound
basis the natural quantity to be considered is the characteristic time derived from the
decay of the mean auto-correlation function.
The temporal correlation function $C(n)$ is typically of the form $\exp(-\alpha n)$ and
the decay rate is given by the parameter $\alpha$.  
The natural way to define a characteristic time $T$ to achieve equilibrium is by setting
$T=1/\alpha$.
In order to find the characteristic decay rates one should consider the
recursive application of the mean matrix $\bm{M}$ on the initial population:
$\bm{Z}_0^{\mathrm{t}}\,\bm{M}^n \bm{Z}_0$.
In fact, it is enough to consider the canonical initial population
$\bm{Z}_{0}=\bm{e}_R=(0,0,...,1)$. 
By direct inspection it is easily verified that the decay of correlations is typically
exponential and given by 
\[
 C(n)=\exp\big(-\log(R(1-d))\,n\big) \,,
\]
where $m=R(1-d)$ is the malthusian parameter.
The decay rate is therefore given by $\alpha=\log(R(1-d))$.

Among others, one possible application of this result relates to the very initial phase
of the infection process. If we consider that during this phase the host immune
system has not been yet stimulated against the virus, one can assume that
the deleterious effects would be solely represented by the viral intrinsic mutation
rates. Therefore, the largest the value of $R$, i.e., the largest the replication
capacity of the initial viral particle the fastest the progeny auto-correlation decays
and reaches equilibrium stabilizing the viral population; intuitively the parameter $R$
defines the degree of virulence of the infection during the early stage of the infective
process. The increment of deleterious effects plays an opposite role on the decay rates.
In fact, as it will be shown below the closest the parameter $d$ is to its critical
value $d_c$ more time is needed to achieve equilibrium.

\subsubsection{The Dynamical Stationary State.}

When the simple phenotypic model is super-critical and is initialized with 
exactly one particle in the class $r$ ($Z_0^r=1$) the effective malthusian parameter is
$m_{\mathrm{e}}=\lambda_r=rc=r(1-d)$ with corresponding normalized right eigenvector
$\bm{u}(r)=\big(u_0(r),\ldots,u_r(r),0,\ldots,0\big)$, where the components $u_k(r)$,
with $k=0,\ldots,r$, are
\begin{equation} \label{eq:NREIGEN}
 u_k(r)=\mathrm{binom}(k;r,1-d)={r \choose k} \, (1-d)^k \, d^{r-k} \,.
\end{equation}

Therefore, the simple phenotypic model has $R-1$ distinct asymptotic distributions of
types of particles, describing $R-1$ distinct \emph{dynamical stationary states},
characterized by their \emph{asymptotic distribution of classes} given 
by~\eqref{eq:NREIGEN} (up to a random scalar perturbation), each one of these being
achieved when the branching process is initialized with exactly one particle in the
class $r$ ($Z_0^r=1$) for $r=2,\ldots,R$, respectively.
Note that when $r=0,1$ the process is always sub-critical.

\begin{theorem} \label{thm:STATIONARYSTATE}
If the simple phenotypic model is super-critical with malthusian parameter $m=R(1-d)$ 
and starts with at least one particle of class $R$ then, in the long run, the relative
number of particles in each class reaches a stable stationary dynamical state and is
(up to a random scalar perturbation) distributed according to the 
\emph{Binomial Distribution:} $\mathrm{binom}(k;R,1-d)$, where $k=0,\ldots,R$ are the
replication classes.
\end{theorem}
\begin{proof}
This is consequence of the generalized Kesten-Stigum~\cite{KS66a,KS66b,KS67} results about
the asymptotic behaviour of decomposable (i.e., with non-primitive mean matrix)
super-critical multitype branching processes and the computation of the normalized right
eigenvector associated to the malthusian parameter $m=R(1-d)$ given by
equation~(\ref{eq:NREIGEN}).
\end{proof}
From theorem~\ref{thm:STATIONARYSTATE} we immediately obtain:
\begin{itemize}
\item The \emph{mean replication capacity} is
      \[
       \Expec(\bm{u})=R(1-d) \,.
      \]
\item The \emph{phenotypic diversity} is
      \[
       \Varia(\bm{u})=Rd(1-d) \,.
      \]
\end{itemize}
It is well accepted that the phenotypic diversity is an important characteristic of the
viral population intuitively related to the idea of population robustness~\cite{EWOL07}. 
In fact, a homogeneous population would be less flexible from the point of view of
adaptation. The variance associated with the stationary state can be understood
as a natural quantity to measure diversity. It shows that the maximum value of the
phenotypic diversity $r/4$ is reached if $d=1/2$ for any value of $r$.
If $R>2$ the variation of the phenotypic diversity as a function of $d$ shows that there
are two different domains to be considered: below $d=1/2$ the diversity is an increasing
function of $d$. It implies that if the population has a typical value of $d<1/2$ the
effect of inducing an increment of $d$ (for instance using mutagenic drugs) increases
the phenotypic diversity. For $1/2<d<d_c$ this effect reverses and diversity decreases
with increasing $d$. This result raises the question if in normal conditions the viral
population adapt to the host environment guided by a principle of maximum phenotypic
diversity or if the environmental conditions simply contribute to fix one possible value
of diversity for the population that may vary from one to another host organism.
Interesting enough, the natural deleterious mutations has been measured for certain
viruses and, as shown in the Table~\ref{tab:COMP}, they are close to the value $d=1/2$.
In the first case one could preview that the set point of the viral disease should be
invariant (or with small variation) for all hosts. On the other hand the second
hypothesis leads to the idea of different responses to treatment depending on the
initial value of $d$ before the adoption of treatment strategies to improve $d$.
At the present the two scenarios may apply to different type of viruses and this point
clearly has to be decided experimentally.

\renewcommand{\arraystretch}{1.5}

\begin{table}[ph]
\tbl{\label{tab:COMP}Experimental results of deleterious mutation rates: (VSV) vesicular
  stomatitis virus, (TEV) Tobacco etch virus and  ($\Phi$X174, Q$\beta$) bacterial
  viruses.}
{\centering
\begin{tabular}{c@{\quad}c@{\quad}c@{\quad}c} \hline
 Virus & $U_d$ & $(1-d)=\mathrm{e}^{-U_d}$ & REF. \\ \hline
 VSV & 0.692 & 0.500 & \cite{SME04} \\
 TEV & 0.773 & 0.461 & \cite{CIE07} \\
 $\Phi$X174 & 0.72 - 0.77 & 0.48 - 0.46 & \cite{CCS09} \\
 Q$\beta$ & 0.74 - 0.86 & 0.47 - 0.42 & \cite{CCS09} \\ \hline
\end{tabular}}
\end{table}

\renewcommand{\arraystretch}{1}

Another important consequence of the above results concerns the efficiency of the use of
mutagenic drugs. In the region $d<1/(R+1)<1/2$ the viral population's most
representative particle is the fittest one (class $R$). If we assume that the drug
action is deeply influenced by drug transport coefficients in different host tissues,
it is important to be assured that local drug concentrations will still eliminate the
set of class $R$ particles. If $d$ increases beyond $1/(R+1)$ the representative
particle of the population is not anymore the fittest one but a set of particles from
different replication classes. Therefore the main drug target represents a group of
average replicating particles of a population with higher phenotypic diversity in which
resistance drug mutants can be contained. In this case one would say that the viral
population displays a kind of endogenous strategy to scape the deleterious action of the
mutagenic drug.
If we assume that deleterious effects are small in the early stage of the infection
process we should expect that at this stage the drug efficiency would be maximum
reinforcing the successful practice of post exposure therapy, currently adopted in the
case of HIV infections\cite{KG97}.

\subsection{The Critical Case: Extinction Threshold}

The clearest way to characterize the time behavior of the viral population at or around
the critical point is through the typical time $T$ to approach equilibrium
derived from the decay of correlations described above.

The expression $T=1/\log(R(1-d))$ shows that at the critical point the equilibrium state
is never reached, i.e., the decay to equilibrium is at least non-exponential. 
A scaling exponent characterizing the behavior of $T$ in the neighborhood of the
critical point $d_{c}$ can be easily obtained.
The expansion around $d_{c}=(R-1)/R$ gives 
\[
 T \approx (1-d_{c})\,|d-d_c|^{-1} \,.
\]
Although it is always possible to calculate intermediate distributions of progeny,
it is quite easy to see that at the critical point the time evolution of densities never
achieves an invariant density.

Unlike in the super-critical regime, the relative number of particles in each
class/sub-population is never stable.
Nevertheless, our preliminary results concerning the dynamics of fluctuations show that
the time variation of the numbers of particles in each separated class follows a pattern
such that the variation observed in one class is rigorously the same observed in all the
others.
This indicate a high level of correlation between the classes in complete analogy with
critical phenomena of many physical systems. We conjecture that in the critical regime
the highly correlated classes in the population behave as an inseparable whole such that
the notion of the population divided in separated classes becomes meaningless.
In other words the correlation between classes makes them behave as if they constitute
one unique class, which reminds one of the basic properties of the error threshold in
Eigen's theory\cite{E71}.
In fact, according to Eigen, when mutational rates are increased beyond a threshold,
infinite viral populations are not anymore able to retain its best adapted variants.
At this critical mutation level, selection is overruled by mutation and all variants
share the same fitness status.
Moreover, populations at Eigen's error threshold do not become extinct, but well defined
replication classes cease to exit, as particles hazardously wander through the
surface of a flat landscape.

If in the super-critical case the notion of the mean replication capacity and therefore
that of the ``mean viral particle'' exists defining a typical scale in the system, in
the critical case this notion is absent. Therefore, in using branching processes to
model the time behavior of viral populations the concept of error threshold should be
identified with that of criticality. In the same direction of reasoning, in terms of
branching processes the existence of lethal mutagenesis should be identified with that
of critical behavior of the model.

The critical behavior of the model can also be observed through the survival 
probability function $\omega(d)$ for $d \lesssim d_c$ as show in Figure~\ref{fig:BASIC2}.
Expansion of the survival probability function $\omega(d)$ around the critical 
point $d_c$ using the system of equations for the extinction
probabilities~(\ref{eq:SYSEQ}) gives directly
\[
 \omega \approx 2 \frac{R}{d_c}\, |d-d_c| \,.
\]
It is interesting to note that the critical exponents of $T$ and $\omega$ are the
same found in critical behavior of a large class of dynamical random networks\cite{RW11}. 
However, it is noteworthy that here we talk about criticality of a process taking place
in time, and therefore the term \emph{critical phenomenon} (imported from equilibrium
statistical mechanics of space distributed systems) is used to highlight the fact that
the survival probability behaves like an order parameter and the amount of deleterious
effects quantified by $d$ behaves as a control parameter that can be changed by external
means.

This fact is reminiscent from the deep relation existing between branching process and 
random network theory, where the survival probability function of a branching process is
identified with the order parameter associated to the emergence of the giant cluster in
a dynamical random network. 
This fundamental observation goes back to Karp\cite{K90} and more recently it has become
the central technique in the study of more general models of random networks\cite{BR08}.
In this direction, it is worth to also note that there is a correspondence between Eigen's
model of molecular evolution and the equilibrium statistical mechanics of an inhomogeneous
Ising system\cite{L87}, again an indication of a relation between statistical mechanics,
random networks and branching processes.
The relation between these theories is certainly expected to bring important new insights
to virus evolution in the future.

\section{Conclusions and Outlook}
\label{sec:CO}

Using the previous theoretical model for virus evolution proposed by 
L\'azaro \etal\cite{LEDM02} and Aguirre \etal\cite{ALM09} as a starting point we
show that virus evolution can be described by an exact solvable multivariate branching
process.
By applying our approach we are able to identify crucial aspects of the dynamics of
replicating viral populations on a sound theoretical basis.
Among these several aspects we are able to demonstrate that -- as long as the beneficial
effects are close to zero -- the two main driving features of a virus population are the
maximum replication capacity and the fraction of the population not affected by
deleterious effects.
Based on this result we show that, as proposed by Bull~\etal, if the product between the
above mentioned parameters $m=R(1-d)$ yields a value less than one the population
undergoes extinction.
On the other hand, if $m=R(1-d)$ is greater than one and the environment is constant
we show that the population will reach an asymptotic stationary state characterized by
the stability of the replicative classes. 
However, the time to reach the stationary equilibrium strictly depends on how 
intense is the deleterious effect, more precisely, the higher $d$ the longer is
the transient phase and when $d$ approaches its critical value $d_c$ the transient 
tends to infinity.

According to our explicit formulas for the progeny distribution, we demonstrate that
virus populations maximize their phenotypic diversity by replicating with $d$ near
$1/2$, for any value of $R$. 
We speculate that this might be a universal property for RNA viruses that replicate
under high mutational rates. In this way by increasing their phenotypic diversity
viruses augments their chances of survival escaping and adapting to environmental
pressures. Maintenance of high mutation rates makes it difficult for a population to
retain their best replicative classes. As a consequence, the best adapted classes are not
the most represented ones in the population, thus not characterizing a classical
Darwinian evolution process.
As far as branching process modeling of viral evolution is concerned its critical
behavior partially resembles the concept of error threshold in Eigen's theory of
molecular quasispecies.
In this regime the replicative classes lose their independence in the sense that they
become so much correlated that the whole set of classes behaves as a single one.

We also demonstrate that by keeping the deleterious effects constant the survival
probability of a virus population will depend on its initial population size.
By increasing the population size at time zero we push the survival probability curves,
in the region before the critical point, towards one (see
Figure~\ref{fig:BASIC2} and equations (\ref{eq:SURVPROB1}), (\ref{eq:SURVPROB2})).
According to this result it can be speculated that virus with greater innoculums have a
better chance of survival colonizing new hosts. 
Interestingly enough and in a frontal disagreement to the above observations it has been
shown that only a limited number of particles, and in some cases even one particle, is
enough to start a new infectious process in a host\cite{K08,ZDE11}.
However, according to the model and as discussed before, the $R$ parameter determines
the success of an incoming virus population because the corresponding value of $d_c$ is
uniquely given by $R$. The present work suggests that minimum innoculums must have at
least one particle with replicative capacity large enough in order to survive in the new
host. We speculate that those particles with maximum replicative capacity should
constitute the effective innoculum described in Zwart \etal\cite{ZDE11}.
In fact, the experimental data about viral load in HIV early infected patients strongly
suggests that the host deleterious effects over the viral population are minimal and
increase after the onset of the immunological response\cite{RQCLSP10}.
We note that the characteristic form of this data can be easily reproduced by the model
(see Castro \etal\cite{D11,CABJ11}). 

Finally, it is important to mention that the close relation of the theory of branching
processes (as used in the present work) and dynamical Erd\"os-Renyi type networks
indicates that the latter may be brought to bear in the modeling of virus populations. 
The relation between these two theories is undoubtedly a research avenue with promising
potential to improve our knowledge of the dynamical laws governing the evolution and
adaptation of viral populations. 

\section*{Acknowledgments}
FA wish to acknowledge the support of CNPq through the grant PQ-313224/2009-9 and
thanks FAP-UNIFESP and BIOMAT Consortium for the financial support to present this 
work at the ``12th International Symposium on Mathematical and Computational Biology''.
FB recieved support from the Brazilian agency FAPESP. 
DC received financial support from the Brazilian agency CAPES. 

\appendix
\section{Review of the Theory of Branching Processes} 
\label{sec:BRTBP}

In this section we collect a few definitions and results from the theory of branching
process that will be necessary in our analysis of the phenotypic model.

\subsection*{The Mean Matrix of a Branching Process}

Consider a multitype branching process $\bm{Z}_n$ with offspring probability distribution
$\bm{\zeta}$ and probability generating function $\bm{f}$.
Suppose that $\bm{\zeta}$ has all its first moments finite and not all zero.
Then conditioning on the elementary initial populations $Z^r_0=1$ on may define the
\emph{mean matrix} $\bm{M}=\{M_{ij}\}$ of the multitype branching process $\bm{Z}_n$ by
\[
 M_{ij}=\Expec(Z_1^i|Z_0^j=1) \quad\forall\, i,j=0,\ldots,R \,.
\]
In general, a multitype Galton-Watson branching process can be classified into
\emph{decomposable} and \emph{indecomposable} according to which its mean matrix
is reducible or irreducible, respectively.
A non-negative matrix $\bm{M}=\{M_{ij}\}$ $(0 \leqslant i,j \leqslant R)$
is called \emph{irreducible} if for every pair of indices $i$ and $j$, there
exists a natural number $n$ such that $\big(\bm{M}^n\big)_{ij}>0$ and it is
called \emph{reducible} otherwise (see Gantmatcher\cite{G05}).
There is another characterization of irreducibility in terms of the graph of the
matrix.

The \emph{graph} $\mathcal{G}(\bm{M})$ of $\bm{M}$ is defined to be the directed
graph on $R$ nodes $\{0,1,\ldots,R\}$, each corresponding to a type of particle,
in which there is a directed edge leading from node $i$ to node $j$ if and only if
$M_{ij}\neq 0$.
A graph $\mathcal{G}(\bm{M})$ is called \emph{path connected} if for each pair
of nodes $(i,j)$ there is a sequence of directed edges leading from $i$ to $j$.
A matrix $\bm{M}$ is irreducible if and only if $\mathcal{G}(\bm{M})$ is path
connected (see Meyer\cite{M00}).

A multitype Galton-Watson branching process is called \emph{positively regular}
if its mean matrix $\bm{M}$ is \emph{primitive}, that is, $\bm{M}^n$ is positive
for some positive integer $n$.
In particular, a positively regular branching process is indecomposable, since
a primitive matrix is irreducible (see Gantmatcher\cite{G05} or Meyer\cite{M00}).
Positive regularity is a standard assumption in the study of multitype branching
processes, as it opens up the way to apply the powerful Perron-Frobenius theory
(see Harris\cite{H63} or Athreya and Ney\cite{AN72}).

\begin{example} \sl
The classification of the phenotypic model according to the irreducibility or
reducibility of its mean matrix is the following:
\begin{enumerate}[(i)]
\item In the version ``with zero class'' the mean matrix \eqref{eq:MEANGENERAL} or
      \eqref{eq:MEANBASIC} will have the first column filled with zeros,
      that is, they are not primitive matrices and thus the corresponding branching
      processes are not positively regular.
      Moreover, a quick look at the graph $\mathcal{G}(\bm{M})$ in 
      Figure~\ref{fig:GRAPH} (b) shows that the process is decomposable since the node
      corresponding to particles of type $0$ does not have a direct arrow leading to
      other nodes. 
      In the case of the simple phenotypic model, the corresponding
      graph $\mathcal{G}(\bm{M})$ is shown Figure~\ref{fig:GRAPH} (a).
      Note that there are no dotted arrows since the probability of beneficial
      effects is $0$ and so the graph is \emph{totally path disconnected}, in other
      words, each ``path component''of the graph consists of exactly one node.
\item In the version ``without zero class'' the mean matrix of both models can be
      obtained from \eqref{eq:MEANGENERAL} and \eqref{eq:MEANBASIC} by removing the
      first row and the first column.
      Now the general phenotypic model becomes positively regular, since
      the node corresponding to particles of class $0$ no longer exists.
      The simple phenotypic model still is decomposable, even without the
      node corresponding to particles of class $0$.
\end{enumerate}
\end{example}

\subsection*{Malthusian Parameter and Extinction Probability}

Let $\varrho(\bm{M})$ denote the \emph{spectral radius} of $\bm{M}$,
that is, if $\lambda_1,\ldots,\lambda_R$ are the eigenvalues of $\bm{M}$
then
\[
 \varrho(\bm{M})=\max\big\{|\lambda_r|\big\} \,.
\]
Since $\bm{M}$ is a non-negative matrix, it has at least one largest non-negative
eigenvalue which coincides with its spectral radius (see Gantmatcher\cite{G05} or
Meyer\cite{M00}).
When the largest eigenvalue is positive we shall call it, following Kimmel and
Axelrod\cite{K02}, the \emph{malthusian parameter} $m$ of the branching process
(see also Jagers \etal\cite{JKS07}). 

The malthusian parameter of a multitype Galton-Watson branching process
plays the same role as the mean of the probability distribution of the
offspring in a simple Galton-Watson process and its name is motivated by equation
\eqref{eq:MEANEVOL1}, which implies that $\varrho(\bm{M}^n)=m^n$, the average
population size increases or decreases at a geometric rate, in accordance with the
``Malthusian Law of Growth''.

Finally, it follows from the theory of non-negative matrices that there is a
\emph{left non-negative eigenvector} $\bm{v}$ and a
\emph{right non-negative eigenvector} $\bm{u}$ corresponding to the
eigenvalue $m$:
\[
 \bm{v}^{\mathrm{t}}\,\bm{M}=m\,\bm{v}^{\mathrm{t}}
 \qquad\text{and}\qquad
 \bm{M}\,\bm{u}=m\,\bm{u} \,,
\]
which can be normalized so that
\[
 \bm{v}^{\mathrm{t}}\bm{u}=1 
 \quad\text{and}\quad
 \bm{1}^{\mathrm{t}}\bm{u}=1 \,,
\]
where $\bm{v}^{\mathrm{t}}$ is the transposed of the vector $\bm{v}$.
Moreover, when $\bm{M}$ is irreducible the left and right eigenvectors
are positive (see Gantmatcher~\cite{G05} or Meyer~\cite{M00}).

Let $\bm{\gamma}=(\gamma_0,\ldots,\gamma_R)$ be the
\emph{vector of extinction probabilities}
\[
 \gamma_r=\Prob(\bm{Z}_n=0 \;\text{for some $n$}|Z_0^r=1) \,,
\]
the probability that the process eventually become extinct given that initially there
is exactly one particle of class $r$.
In general, when the initial condition is given by a vector of non-negative integers
$\bm{Z}_0=(Z_0^0,Z_0^1,\ldots,Z_0^R)$ the extinction probability is
\[
 \Prob(\bm{Z}_n=0 \;\text{for some $n$}|\bm{Z}_0)=\prod_{i=0}^R \,\gamma_i{}^{Z_0^i} \,.
\]

A basic result of the theory of branching processes is that the vector of extinction
probabilities $\bm{\gamma}$ is the solution in $[0,1]^R$ with smallest components of
the equation 
\begin{equation} \label{eq:FIXEDPOINT}
 \bm{f}(\bm{\gamma})=\bm{\gamma} \,,
\end{equation}
where $\bm{f}$ is the probability generating function.
Observe that $\bm{1}$ is always a fixed point of $\bm{f}$,
that is, a solution of equation \eqref{eq:FIXEDPOINT}.
Therefore, if there is no other solution of equation \eqref{eq:FIXEDPOINT}
in the unit cube $[0,1]^R$ then the process always has probability $1$ to become extinct.

The main classification result in the indecomposable case, states that there are only
three possible regimes (see Harris\cite{H63} or Athreya and Ney\cite{AN72}):
\begin{enumerate}[(i)]
\item If $m>1$ then $\bm{0}\leqslant\bm{\gamma}<\bm{1}$ is the unique stable
      fixed point of $\bm{f}$ in the unit cube $[0,1]^R$ different than $\bm{1}$
      and the branching process is called \emph{super-critical}.
      Therefore, with positive probability, the population will survive indefinitely.
\item If $m<1$ then $\bm{\gamma}=\bm{1}$ is the unique stable fixed point of
      $\bm{f}$ in the unit cube $[0,1]^R$ and the branching process is
      called \emph{sub-critical}.
      Therefore, with probability $1$, the process will become extinct in finite time.
\item If $m=1$ then $\bm{\gamma}=\bm{1}$ is the unique marginal fixed point of $\bm{f}$
      in the unit cube $[0,1]^R$ and the branching process is called \emph{critical}.
      Here, the expected time to extinction is infinite, despite the fact that
      extinction is bound to occur almost surely.      
\end{enumerate}

Unfortunately this theorem does not cover all the interesting cases, one important
example for us being the phenotypic model for viral evolution.
Nevertheless, one of the earliest results about decomposable branching processes is the
generalization of the classification, due to Sevastyanov (see Harris\cite{H63} and
Ji\v{r}ina\cite{J70}).
In the general decomposable case, there is a fourth alternative identified by
Sevastyanov\cite{H63,J70} and in order to formulate this condition we need to
introduce another important concept.

A multitype Galton-Watson branching process is called \emph{singular}
if its probability generating function is linear without constant term, that is,
$\bm{f}(\bm{z})=\bm{M}\bm{z}$.
In this case, there is no branching since each particle produces exactly one particle
that can be of any class and the process is equivalent to an ordinary finite Markov 
chain.
More generally, a decomposable process may have \emph{singular path components}.
Two nodes $i$ and $j$ are said to be in same \emph{path component} if there is a
sequence of directed edges leading from $i$ to $j$ and a sequence of directed
edges leading from $j$ to $i$.
This procedure defines a partition of the set of nodes into equivalence classes,
called \emph{path components} of the graph $\mathcal{G}(\bm{M})$.
We say that a path component $C$ of $\mathcal{G}(\bm{M})$ is a
\emph{singular path component}, if any particle whose class is in $C$ has probability
$1$ of producing, in the next generation, exactly one particle whose class is in $C$.
Equivalently, the component functions of the probability generating function
corresponding to the classes in a path component $C$ are linear functions of the
variables corresponding to the classes in the path component $C$.
In other words, the ``part'' of the probability generating function corresponding to the
classes in $C$ is that of a singular branching process. 
The existence of singular components is obviously an obstruction to extinction,
for instance, in a decomposable singular process all path components are singular.
In fact, the result of Sevastyanov states that if there is at least one
\emph{singular path component} then the branching process never become extinct, no matter
what is the value of the malthusian parameter.

\begin{example} \sl
The graph corresponding to the general phenotypic model (Figure~\ref{fig:GRAPH} (b))
have two path components: $\{0\}$ and $\{1,2,3,\ldots,R\}$.
In the simple phenotypic model (Figure~\ref{fig:GRAPH} (a)), the path components are
exactly the sets containing one node, $\{0\}$, $\{1\}$, \ldots, $\{R\}$.
From the expressions of the generating functions \eqref{EQ:genfunc1} and
\eqref{EQ:genfunc2} it is clear that there are no singular path components in any
of the models -- simple or general, ``with zero class'' or ``without zero class''.
Moreover, the general phenotypic model ``without zero class'' is 
\emph{positively regular}.
Therefore, the phenotypic model displays only the three regimes determined by the
malthusian parameter, which depends on the values of the parameters $b,c,d$ and $R$. 
\end{example}

It is important to stress that the regime of a multitype branching process can not
be read from the mean matrix alone (i.e, the malthusian parameter).
Essentially this happens because of the existence of decomposable branching processes
with singular components.

\begin{example} \sl
Consider the following generating functions:
\[
\begin{split}
 \bm{g}(z,w) & = \big(1/2+1/2z^2,(dz+cw)^2\big) \,, \\
 \bm{h}(z,w) & = \big(z,(dz+cw)^2\big) \,,
\end{split}
\]
where $0<c,d<1$ and $c+d=1$.
They have the same mean matrix given by
$\bm{M}=\bigl(\begin{smallmatrix}
 1 & 2d \\
 0 & 2c \\
\end{smallmatrix}\bigr)$ and so the malthusian parameter is $m=\max\{1,2c\}$.
It is easy to solve the fixed point equation \eqref{eq:FIXEDPOINT}
in both cases and compute the respective extinction probability vectors 
$(\gamma_1,\gamma_2)$: for the function $\bm{g}$ we have that
$\gamma_1=1$ and $\gamma_2=d^2/c^2$ if $0 \leqslant d \leqslant\tfrac{1}{2}$
and $\gamma_2=1$ if $\tfrac{1}{2} \leqslant d \leqslant 1$.
For the function $\bm{h}$ we have that $\gamma_1=\gamma_2=0$.
Therefore, the branching process defined by $\bm{g}$ becomes extinct if and only if
$c\leqslant 1/2$ while the branching process defined by $\bm{h}$ never becomes extinct
irrespective of the value of the malthusian parameter!
\end{example}

\subsection*{Asymptotic Behaviour of Surviving Populations}

According to the ``Malthusian Law of Growth'' it is expected that a super-critical
branching process will grow indefinitely at a geometric rate proportional to $m^n$
and we would like to write $\bm{Z}_n \approx m^n \,\bm{W}_n$, where $\bm{W}_n$
is a random vector with a finite ``asymptotic distribution of classes'' when
$n \to \infty$.
The formalization of this heuristic argument is due to Kesten and Stigum
(see Kesten-Stigum\cite{KS66a,KS66b} for the case of indecomposable multitype branching
processes and Kesten-Stigum\cite{KS67} for the case of a general decomposable multitype
branching processes).

Let us first recall the result in the indecomposable case (see Athreya and Ney\cite{AN72}). 
Consider a super-critical branching process with $m>1$ and suppose that
the vector valued random variable $\bm{\zeta}$ satisfies a technical condition called
\emph{Kesten-Stigum ``$\zeta\log\zeta$'' condition} (see Lyons \etal\cite{LPP95} and
Olofsson\cite{O98}), which is always satisfied in our case, since the probability
distribution of the offsprings has finite support.
It is natural to define the normalized random vector $\bm{W}_n=\bm{Z}_n/m^n$.
This normalized random vector has a limit when $n \to \infty$, that is,
there exists a scalar random variable $W \neq 0$ such that, with probability one,
\[
 \lim_{n\to\infty} \bm{W}_n=W \,\bm{u} \,,
\]
where $\bm{u}$ is the normalized right eigenvector corresponding to the malthusian
parameter $m$ ($\bm{1}^{\mathrm{t}}\bm{u}=1$) and
\[
 \Expec(W|\bm{Z}_0)=\bm{v}^{\mathrm{t}} \bm{Z}_0
\]
where $\bm{v}$ is the normalized left eigenvector corresponding to the malthusian
parameter $m$ ($\bm{v}^{\mathrm{t}}\bm{u}=1$).

\pagebreak

An important step in the proof of Kesten-Stigum theorem is the Kurtz\cite{KLPP94}
\emph{convergence of classes} theorem:
\begin{equation} \label{eq:KURTZ}
 \lim_{n\to\infty} \,\dfrac{\bm{Z}_n}{|\bm{Z}_n|}=\bm{u} 
 \qquad\text{(almost surely).}
\end{equation}
Combining this with the Perron-Frobenius theorem (see Meyer\cite{M00}) one obtains
\begin{equation} \label{eq:KPF}
 \lim_{n\to\infty} \,\dfrac{\bm{Z}_n}{|\bm{Z}_n|}=
 \lim_{n\to\infty} \,\dfrac{\langle\bm{Z}_n\rangle}{|\langle\bm{Z}_n\rangle|}=
 \bm{u} \,,
\end{equation}
where the convergence in the first limit is in probability.
The approach adopted in Cuesta \etal\cite{CACM11,CCMA11} relates only to the second
equality involving the limit of mean values of equation \eqref{eq:KPF}.
By explicitly considering the microscopic model as a multivariate branching process the
equality of the two limits in equation \eqref{eq:KPF} is guaranteed.
This result may be useful for computational simulations of the model, since one
may compute the eigenvector $\bm{u}$ by sampling the population and taking
averages.

The meaning of the Kesten-Stigum theorem is that the total size of the population
divided by $m^n$, converges to a random vector, but the relative proportions of the
various ``classes'' approach fixed limits.
Since we are assuming that the process is indecomposable the normalized right
eigenvector $\bm{u}=(u_0,\ldots,u_R)$ is positive and is normalized so that $\sum_r
u_r=1$, therefore it defines a probability distribution on the set of classes
$\{0,\ldots,R\}$.
It is called the \emph{asymptotic distribution of classes} of the multitype branching
process. 

In order to extend these results to the case where the branching process is decomposable
one should employ the \emph{Frobenius normal form} of the mean matrix $\bm{M}$, which
is reducible in this case (see Gantmatcher\cite{G05}).
Kesten and Stigum\cite{KS67} show that it is possible, by rearranging the rows and
columns, to rewrite the mean matrix in a block upper triangular form in such a way
that the diagonal blocks are irreducible square matrices associated to components
of the decomposable branching process.
By a \emph{component} of a decomposable branching process we mean a subset of classes 
such that their associated nodes in the graph $\mathcal{G}(\bm{M})$ forms a path
component.
Let $\{C_k~:~0 \leqslant k \leqslant N\}$ be the set of components of
$\mathcal{G}(\bm{M})$ ordered according to which $C_k\prec C_l$ if there is a sequence
of directed edges leading from some $i \in C_k$ to some $j \in C_l$.
Given two components $C_k$ and $C_l$ define the sub-matrix
\[
 \bm{M}(k,l)=M_{ij} \quad\text{with}\quad i\in C_k \,, j\in C_l \,.
\]
Then, for each $k$, the square sub-matrix $\bm{M}(k)=\bm{M}(k,k)$ is the irreducible
mean matrix of the sub-process
\[
 \bm{Z}_n(k)=\{Z_n^i~:~i\in C_k\} \,.
\]
Now the order of the components $C_k$ allows us to rearrange the rows and columns of
$\bm{M}$ in such a way that
\[
 \bm{M}=\begin{pmatrix}
 \bm{M}(0) & \bm{M}(0,1) & \bm{M}(0,2) & \cdots & \bm{M}(0,N) \\
 0         & \bm{M}(1)   & \bm{M}(1,2) & \cdots & \bm{M}(1,N) \\
 0         & 0           & \bm{M}(2)   & \cdots & \bm{M}(2,N) \\
 \vdots    & \vdots      & \vdots      & \ddots & \vdots \\
 0         & 0           & 0           & 0      & \bm{M}(N)
\end{pmatrix}
\]
Therefore, the sub-process $\bm{Z}_n(k)$ ``receives input'' from the sub-process
$\bm{Z}_n(l)$, with $k<l$, throughout the sub-matrix $\bm{M}(k,l)$.
Note that if the sub-matrices $\bm{M}(k,l)$ are all zero then the branching
process splits as a sum of independent indecomposable branching processes.

\begin{example} \sl
The matrices \eqref{eq:MEANBASIC} and \eqref{eq:MEANGENERAL} of the simple
and the general phenotypic models, respectively, already are in the normal form:
\begin{enumerate}[(i)]
\item In the simple phenotypic model we have that $N=R$, with one-dimensional
      diagonal sub-matrices $\bm{M}(k)=\big(k(1-d)\big)$, with one-dimensional
      sub-matrices $\bm{M}(k,k+1)=\big((k+1)d\big)$ and one-dimensional sub-matrices
      $\bm{M}(k,l)=0$ if $l>k+1$.
\item In the general phenotypic model we have $N=1$, with the first diagonal
      sub-matrix $\bm{M}(0)=\big(0\big)$ (or $\bm{M}(0)=\big(1 \big)$ for the first
      variation of the model), the second diagonal sub-matrix $\bm{M}(1)=M_{ij}$, 
      with $i,j=1,\ldots,R$ and $\bm{M}(0,1)=\big(d~0 \cdots 0)$.
\end{enumerate}
\end{example}

Now observe that if $Z_0^i=1$ with $i\in C_k$ then for $l>k$, the sub-process
$\bm{Z}_n(l)=\bm{0}$ for all $n\geqslant 0$.
That is, the branching process behaves as if the sub-processes $\bm{Z}_n(l)$
for all $l>k$ did not exist.
Since each non-zero diagonal sub-matrix $\bm{M}(l)$ is irreducible, it has
a largest positive eigenvalue $m(l)$ and then we may define the 
\emph{effective malthusian parameter} of the sequence of sub-processes
$(\bm{Z}_n(0),\ldots,\bm{Z}_n(k))$ to be
\[
 m_{\mathrm{e}}(k)=\max_{l\leqslant k}\{m(l)\} \,.
\]
The simplest case is when all $m(l)$ are simple eigenvalues of their
respective sub-matrices $\bm{M}(l)$ -- they are distinct amongst each other --
 this is exactly the case for matrices \eqref{eq:MEANBASIC} and 
\eqref{eq:MEANGENERAL}.

In Kesten and Stigum\cite{KS67} the result about the asymptotic behaviour
of irreducible super-critical branching process is generalized to the reducible case. 
The main theorem applied to the case where all $m(l)$ are simple eigenvalues of their
respective sub-matrices $\bm{M}(l)$ states that if the effective malthusian
$m_{\mathrm{e}}(k)>1$ and the ``$\zeta\log\zeta$'' condition holds then for the
normalized random vector $\bm{W}_n(k)=\bm{Z}_n/(m_{\mathrm{e}}(k))^n$ there exists a
scalar random variable $W \neq 0$ such that, with probability one,
\[
 \lim_{n\to\infty} \bm{W}_n(k)=W \,\bm{u}(k) \,,
\]
where $\bm{u}(k)$ is the normalized right eigenvector corresponding to the
effective malthusian parameter $m_{\mathrm{e}}(k)$ and
\[
 \Expec(W|\bm{Z}_0)=\bm{v}(k)^{\mathrm{t}}\bm{Z}_0 \,,
\]
where $\bm{v}(k)$ is the left eigenvector corresponding to the effective malthusian
parameter $m_{\mathrm{e}}(k)$.
Moreover, Kurtz's \emph{convergence of classes} theorem~\eqref{eq:KURTZ} still holds.
But one should note that the normalized right and left eigenvectors are not
positive anymore.
In fact, $\bm{v}(k)$ may have negative entries, but only those associated
to the components $C_l$ with $l \leqslant k$, for which $\bm{Z}_0$ is zero.
The right normalized eigenvector is of the form $\bm{u}(k)=(u_0,\ldots,u_r,0,\ldots,0)$,
where $(u_0,\ldots,u_r)$ is the non-negative right normalized eigenvector of the
sub-matrix corresponding to the the sequence of sub-processes
$(\bm{Z}_n(0),\ldots,\bm{Z}_n(k))$, and so is a probability distribution.

\subsection*{Critical Behavior and Regime Transition}

The critical state separates the super-critical and the sub-critical regimes where the
branching process has two distinct behaviors in time and thus characterizes the
existence of regime transition with genuine critical behavior.
In fact, the decay of correlation functions described in the next section
for the case of the simplest model clarifies this point.

\enlargethispage{5mm}

Although in a critical branching process $\bm{Z}_n\to 0$, almost surely, when
$n\to\infty$, one still may obtain a meaningful asymptotic law by conditioning on
non-extinction. 
See Mullikin\cite{M63} and Joffe and Spitzer\cite{JS67} for the indecomposable
case and Foster and Ney\cite{FN78} for certain decomposable cases.

In the indecomposable critical case $\bm{Z}_n$ grows at a linear
rate proportional to $n$ (see Harris\cite{H63} or Athreya and Ney\cite{AN72}),
and so one should consider the normalized random vector $\bm{Y}_n=\bm{Z}_n/n$. 
If the second moments are finite and the branching process is non-singular, there is a
scalar random variable $Y \neq 0$ such that
\[
 \lim_{n\to\infty} \bm{Y}_n=Y \bm{u} \quad\text{given that } \bm{Z}_n\neq 0 \,,
\]
where $\bm{u}$ is the normalized right eigenvector corresponding to the malthusian
parameter $m$ and with convergence only in \emph{distribution}, which is weaker than
the \emph{almost surely convergence} in the super-critical case.

\nocite{*}

\bibliographystyle{plain}
\bibliography{biomat2012}


\end{document}